\documentclass[12pt,3p,times]{elsarticle}

\usepackage{amssymb}
\usepackage{amsmath}
\usepackage[T1]{fontenc}
\usepackage{lmodern}
\usepackage{bbm}
\usepackage{epsfig}
\usepackage{epstopdf}
\usepackage{color}
\usepackage{xurl}

\usepackage{amsthm}
\usepackage{mathtools}
\newtheorem{theorem}{Theorem}
\newtheorem{lemma}{Lemma}

\newtheorem{definition}{Definition}

\usepackage{algorithm2e}

\journal{journal}

\begin{document}

\begin{frontmatter}



\title{Multi-layer diffusion model of photovoltaic installations}


\author{Tomasz Weron}
\ead{tomasz.weron@pwr.edu.pl}
\address{Department of Applied Mathematics, Faculty of Pure and Applied Mathematics,\\ Wroc\l{}aw University of Science and Technology, Wroc\l{}aw, Poland}

\begin{abstract}
Nowadays, harmful effects of climate change are becoming increasingly apparent. A vital issue that must be addressed is the generation of energy from non-renewable and often polluting sources. For this reason, the development of renewable energy sources is of great importance. Unfortunately, too rapid spread of renewables can disrupt stability of the power system and lead to energy blackouts. One should not simply support it, without ensuring sustainability and understanding of the diffusion process. In this research, we propose a new agent-based model of diffusion of photovoltaic panels. It is an extension of the $q$-voter model that utilizes a multi-layer network structure. The novelty is that both opinion dynamics and diffusion of innovation are studied simultaneously on a multidimensional structure. The model is analyzed using Monte Carlo simulations and the mean-field approximation. The impact of parameters and specifications on the basic properties of the model is discussed. Firstly, we show that for a certain range of parameters, innovation always succeeds, regardless of the initial conditions. Secondly, that the mean-field approximation gives qualitatively the same results as computer simulations, even though it does not utilize knowledge of the network structure.
\end{abstract}


\end{frontmatter}


\section{Introduction}
\label{sec:intro}

Nowadays, we are experiencing effects of climate change and environmental pollution more and more often.
What might have once seemed insignificant, such as glaciers melting thousands of kilometers away, now begins to affect us directly. While the gigantic fires that devastated countries like Australia, Canada or Greece have spared the one of the author's origin, Poland, we experience another environmental disaster -- air pollution. In Poland, low-quality heating installations in single-family households or multi-family tenements are among the main causes of its formation~\cite{world2019air}. Generally speaking, the problem is the generation of energy from non-renewable and often polluting sources. High emissions of pollutants directly threaten our lives and, in the already coming future, lead to dangerous climate change \cite{GEE17}, of which global warming is probably the most widely discussed. It not only affects entire ecosystems and the comfort of life, but according to some studies may lead to aggravation and spread of many dangerous diseases as well \cite{MOR22}. 

That is why the development of renewable energy sources (renewables, RES) is so important. Across the European Union in 2023, about $45\%$ of electricity generation came from RES \cite{page_eu1}. Poland, although it saw a rapid expansion of renewables, still falls behind with around $26\%$ \cite{page_pl1}. In Europe, wind is the main source of renewable energy \cite{page_eu2}. Unfortunately, wind farms require huge investments and vast free space, far from residential buildings. An alternative devoid of these limitations, and rapidly gaining popularity in Poland over the past few years, are photovoltaic panels (PVs), see Fig.~\ref{fig:data} for exact numbers. Those can be installed on a roof of a single-family house without adversely affecting the quality of life, while the cost of such a project is achievable already for a middle-income family \cite{page_pl2}. However, it is not only important to provide this opportunity and encourage it through subsidy programs or public campaigns, but to ensure sustainability as well. Excessive support of RES diffusion can lead to a significant increase in the variability of demand for conventional generation on a 24-hour basis, and in extreme cases can even cause loss of stability of the power system. Examples are countries/regions with a large operation of the sun during a year, such as Australia \cite{MAT19} or California \cite{MAI17}, but also Poland's neighbor -- Germany \cite{GOO19}. Apart from that, there is another threat, more closely related to the common folk. Present transmission system grids in Poland and other countries are not prepared for such a rapid PVs expansion. If unrestricted, it may lead to grid overload and blackouts. In fact, this is an already emerging risk \cite{REI23}.

\begin{figure}
    \centering
	\includegraphics[width=0.8\textwidth]{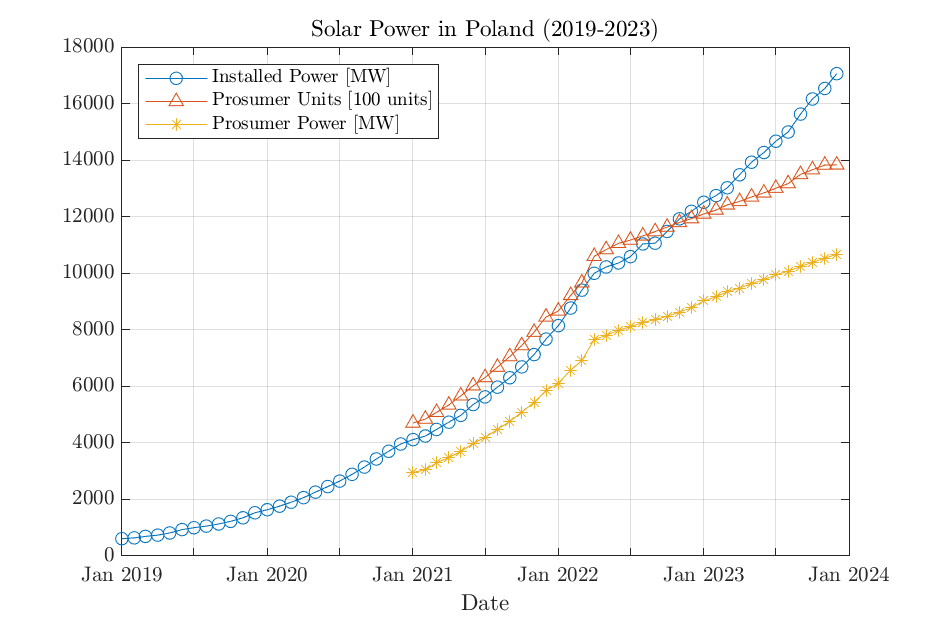}
	\caption{Data on solar installations in Poland in years 2019--2023, according to ARE (Polish energy market agency). More detailed data is being collected from 2021 onward, most likely due to the growing interest in the prosumer market. Data freely available at www.are.waw.pl.}
	\label{fig:data}
\end{figure}

The aim of this research is to better understand how various factors impact the diffusion of PVs. Diffusion is understood here as the process of spreading a new product through a population~\cite{ROG03}. Diffusion of innovation has been modeled for years. It began in a very simplified manner. A classic example is the Bass diffusion model -- a fully deterministic one, consisting of a single differential equation \cite{BAS69}. However, for a problem as complex as the diffusion of PV installations, aggregated models are just not enough \cite{RAH08}. They are unable to describe clustering of individuals, a phenomenon observed in real life~\cite{WOL20}. The aforementioned Bass diffusion of innovation model can dramatically change its behavior, when rewritten to an agent-based model (ABM) and tested on a network structure \cite{RAN11}. In social sciences, an agent-based model is usually understood as a simulation of certain interactions between so-called agents (representing individuals, households, companies, etc.), and taking place within a certain structure, symbolizing a network of acquaintances, contacts or cooperation \cite{MAC02}. Mathematically, we would call such a structure a graph, with vertices being the agents, and edges -- the connections between them~\cite{NEW10}. ABMs allow for a much more accurate representation of reality, including heterogeneity of individuals and the interactions between them \cite{NEW10}. 
As such, ABMs are currently one of the most powerful tools in studying opinion dynamics and diffusion of innovation \cite{KRU17}. Although they have been used for years to model the diffusion of new energy solutions \cite{KOW14,JEN15,BYR16,WER18,VDK19}, their applications in modeling the diffusion of PVs remain few \cite{PAL15,MIT19,CAP20,ZHA22}.

Although much more advanced than simple deterministic models, ABMs are still merely a hypothetical approximation of reality. With that in mind, one should adjust them to a problem at hand as accurately as possible, instead of constructing a one-size-fits-all model. A huge role here is played by the underlying structure. For example, when modeling a spread of gossip in a high school class, a simple network of class acquaintances would be sufficient. However, to properly represent the flow of information and exchange of opinions in the modern world, a much more complex structure is needed. This is where multi-layer networks come in. They are applied in many fields of science \cite{ALE19}, for they can provide multi-level representation of real world dependencies \cite{BIA18}. For instance, an individual (agent) may learn about recent sport results either from friends at work (one layer), or through social media (another layer).  

Sociologists have long pointed out that structures of social interaction should not be reduced to single-layer networks~\cite{BIA18}. However, multi-layer structures have only been studied intensively since the last decade \cite{KIV14}. Recently, they have been used in modeling the diffusion of innovation \cite{XIO18,LI19}. Nevertheless, this is still a relatively fresh concept. One issue that arises with multi-layer networks is generalization of models' rules that where originally implemented on single-layer structures. For instance, one may assume that social influence is only effective if it comes from all the layers (AND rule). However, it can also be assumed that the influence is effective even if it comes from a single layer only (OR rule) \cite{LEE14b}. In this research, we follow the approach from \cite{CHM15} and study both variants.

Monte Carlo (MC) computer simulations are the main research method for ABMs. However, in some limited cases, as for networks being complete or random graphs with a low clustering coefficient, analytical methods such as mean-field (MFA) or pair approximation (PA) can be used \cite{GLE13,JED17,WER24}. Unfortunately, real world networks are characterized by a high clustering coefficient \cite{NEW10}, so often the analytical results obtained this way differ quantitatively from ones given by Monte Carlo simulations \cite{JED17}. 
The need to sweep a multidimensional space of input parameters and perform multiple independent repetitions to obtain good statistical results makes Monte Carlo simulations very time-consuming. For this reason, we utilize both computer simulations and analytical methods in this research. 

A popular agent-based model that already has found its use in studying diffusion of eco-innovation is the $q$-voter model \cite{KOW14,WER18,CAS09b}. In this model, conformity is the basic form of social response. Agents are characterized by a single binary variable denoting their opinion and placed in nodes of some underlying graph structure. Their opinions change upon impact with so-called groups of influence, i.e. $q$ of the agent's neighbors chosen randomly, but only if a given group presents an unanimous opinion. A well-established extension to the $q$-voter model is the addition of independence -- the probability that an agent acts and changes opinion independently of the group of influence \cite{NYC12}. The $q$-voter model has been already examined on multi-layer networks, but to the best of our knowledge only ones consisting of two identical layers \cite{CHM15,GRA20}. 

In this paper, we introduce a new model based on the aforementioned $q$-voter model with independence. Our goal is to capture both opinion dynamics and diffusion of innovation jointly. Hence, we introduce (compared to the original $q$-voter model) a second agent's attribute, in addition to the opinion -- an adoption state. Similarly to the opinion, it is a binary one. Positive value means that an agent possesses a PV installation, negative -- it does not. An agent (in this case a household) can acquire knowledge of PVs from two sources. Either it sees panels on the roofs of neighboring households, or communicates with friends/colleagues. To model this twofold dynamic, we utilize a two-layer network as the underlying topology. The first layer of the network depicts the spacial location of an agent, similarly to \cite{SCH71,KOW18}. From a mathematical point of view it is just a square lattice (SL) with Moore's neighborhood \cite{MOO62}. At this level, each agent possesses only knowledge of the adoption states of the closest neighborhood, and these adoption states shape agent's opinion, according to the $q$-voter rule. The agent cannot share opinion or be subject to the opinions of the neighborhood. This is the visual observation part. A single-family house owner sees only PV panels, or lack thereof, on the roofs of neighboring houses, and does not exchange information or opinions with neighbors. The second layer of the network corresponds to structure of contacts and relationships of agents. For this reason, here we use two-dimensional Watts-Strogatz graph (WS2D) \cite{BAT22}. It possesses some characteristics of real world social networks \cite{NEW10,WAS94} and, in a specific case (when randomness equals $0$), reduces to a square lattice. At this level, a connection between two agents implies their familiarity, through which they can exchange opinions, but do not observe their adoption states. In this case, social influence is also described by the $q$-voter model. Lastly, it is an agent's opinion solely that impacts its adoption state. An agent with a positive opinion can install photovoltaic panels with probability $a_1$, one with a negative one can uninstall with probability $a_2$. The reason for this research is to thoroughly examine this mathematical model, to understand how its parameters shape the outcome.

The remainder of this article is organized as follows. In the next section, we provide a more detailed description of the model and the methods used. Then, in Section~\ref{sec:result}, we examine the two variants, AND and OR rules, which were explained above. Both by means of computer simulations, numerical and analytical methods. Finally, in Section~\ref{sec:conclusion}, we wrap up. Additionally, in the Appendix, we consider various initial conditions.

\section{Model and Methods}
\label{sec:model}

\begin{figure}
    \centering
    \includegraphics[width=\textwidth]{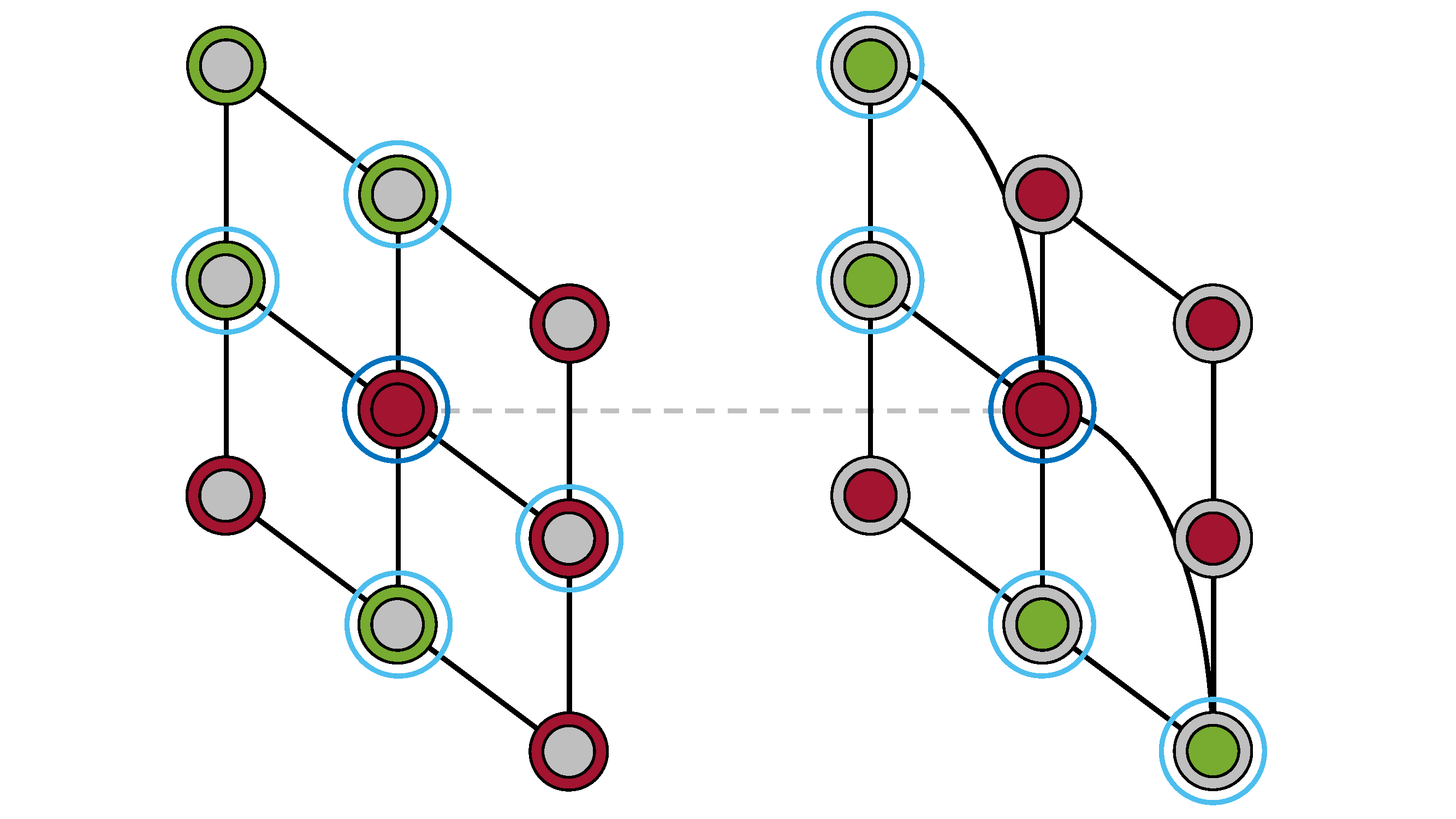}
    \caption{Graphical representation of the model. The network consists of 2 layers: Square Lattice (SL, left side), on which adoption states $A_i$ are visible and two-dimensional Watts-Strogatz (WS2D, right side) with opinions $S_i$. Adoption states $A_i$ are represented by outer circles (green -- $A_i=+1$, red -- $A_i=-1$), while opinions -- by inner circles. Grey areas correspond to adoption states or opinions unknown to the target agent (marked with a dark blue circle). Groups of influence (of size $q=4$, marked with light blue circles) are constructed independently on each layer. In the given example, such a choice would be sufficient to change target's opinion ($S_{target} \to +1$) in the OR variant, but not in the AND variant, as unanimity is only achieved in one of the two groups of influence.}
    \label{fig:model}
\end{figure}

\subsection{Simulation Model}
\label{ssec:model_sim}

We consider a set of $N$ agents, each of which is characterized by 2 binary variables: adoption $A_i = \pm 1$ (adopted or not adopted), and opinion $S_i = \pm 1$ (positive or negative), for $i=1,2,\dots,N$. Agents are located on a two-layer network (Fig.~\ref{fig:model}). The first layer of the network depicts the spacial location of an agent and is represented by a square lattice (SL) with Moore's neighborhood, meaning that each agent possesses 8 neighbors surrounding it (except for agents in the corners and along the edges, which have 3 and 5 neighbors, respectively) \cite{MOO62}. At this level, an agent sees only the adoption states of its neighbors. The second layer corresponds to the structure of social ties of agents, which is described by a two-dimensional Watts-Strogatz graph \cite{BAT22}. Here, an agent sees only the opinions of its neighbors. Each layer is a connected simple graph~\cite{GIB85}, as described in Definition~\ref{def:adjacency}.
\begin{definition}
    Let $G = \left[ G_{X,i,j} \right]$ denote the adjacency matrix for the two-layer network, with each layer $X$, $X \in \{1,2\}$, being a connected simple graph, i.e. an unweighted, undirected graph containing no graph loops or multiple edges~\cite{GIB85}. Then, for $i,j =1,2,\dots,N$:
    \begin{itemize}
        \item $\forall_X \forall_i \forall_j G_{X,i,j} = 1 \,\, \iff$ edge between $i$ and $j$ exists on layer $X$,
        \item $\forall_X \forall_i \forall_j G_{X,i,j} = 0 \,\, \iff$ edge between $i$ and $j$ does not exist on layer $X$,
        \item $\forall_X \forall_i \forall_j G_{X,i,j} = G_{X,j,i}$,
        \item $\forall_X \forall_i G_{X,i,i} = 0$.
    \end{itemize}
    \label{def:adjacency}
\end{definition}

We investigate the model through Monte Carlo simulations with a random sequential updating scheme. Within a single simulation, the time is measured in so-called Monte Carlo steps (MCS)~\cite{BIN10}. One MCS consists of $N$ elementary events, each of length $\Delta t = \frac{1}{N}$. In each elementary event, an agent $i$ is chosen at random (uniformly from all $N$). Then, with probability $p$, agent $i$ acts independently and changes its opinion $S_i$ randomly. With complementary probability $1-p$, the agent is susceptible to social influence (based on the $q$-voter model~\cite{CAS09b}), combined from both layers of the network with respect to the variant, AND or OR. Lastly, if agent $i$ has positive opinion $S_i=+1$, but negative adoption state $A_i=-1$, it changes the latter to positive, $A_i \to +1$, with probability $a_1$. Otherwise, if it possesses negative opinion $S_i=-1$, but positive adoption state $A_i=+1$, it looses the latter, i.e. $A_i \to -1$, with probability $a_2$. Simulation details are shown in Algorithm~\ref{alg:multilayer}, and the graphical representation of the model in Fig.~\ref{fig:model}. In Algorithm~\ref{alg:multilayer}:
\begin{itemize}
    \item $\mathcal{U}\left[0,1\right]$ stands for a continuous uniform distribution, while $\mathcal{U}\{X\}$ -- for a discrete one, where each element from set $X$ is chosen with equal probability.
    \item $G$ is the adjacency matrix for the two-layer network, as per Definition~\ref{def:adjacency},
    \item $p$, $a_1$, and $a_2$ denote probabilities of independence, adopting (getting positive adoption state) and unadopting, respectively. As probabilities, $p$, $a_1$, $a_2 \in [0,1]$. In this research, we only consider $a_1 \in (0,1]$ and:
    \begin{equation}
        a_2 = h \times a_1, \label{eq:h}
    \end{equation}
    where $h \in (0,1)$.
    \item $q$ stands for the size (the number of neighbors) of the group of influence. In here, we consider $q \in \mathbb{N}$, $q \geq 2$.
\end{itemize}
\RestyleAlgo{ruled}
\begin{algorithm}[ht!]
\caption{Simulation dynamics}\label{alg:multilayer}
\For{$t \coloneqq 1$ \KwTo $T$}{
    \For{$k \coloneqq 1$ \KwTo $N$}{
        $i \coloneqq i \sim \mathcal{U}\{1,\dots,N\}$\\
        $r \coloneqq r \sim \mathcal{U}\left[ 0,1 \right]$\\
        \uIf{$r < p$}{
            $r \coloneqq r \sim \mathcal{U}\left[ 0,1 \right]$\\
            \uIf{$r < \frac{1}{2}$}{
                $S_i \coloneqq -S_i$\\
            } 
        }
        \Else{
            \For{$l \coloneqq 1$ \KwTo $q$}{
                $j_{1,l} \coloneqq j_1 \sim \mathcal{U}\{ j_1 : G_{1,i,j_1}=1 \}$\\
                $j_{2,l} \coloneqq j_2 \sim \mathcal{U}\{ j_2 : G_{2,i,j_2}=1 \}$\\
            }
            $Q_1 \coloneqq \frac{1}{q} \sum_l^q A_{j_{1,l}}$\\
            $Q_2 \coloneqq \frac{1}{q} \sum_l^q S_{j_{2,l}}$\\
            \uIf{$\text{Variant} = \text{AND}$}{ 
                \uIf{$Q_1 + Q_2 = -2 S_i$}{
                    $ S_i \coloneqq -S_i$\\
                    }
            }
            \ElseIf{$\text{Variant} = \text{OR}$}{ 
                \uIf{$\left( Q_1 = -S_i \right.$ {\bf and} $\left. Q_2 \neq S_i \right)$ {\bf or} $\left( Q_1 \neq S_i \right.$ {\bf and} $\left. Q_2 = -S_i \right)$}{
                    $S_i \coloneqq -S_i$\\
                }
            }
        }
        $r \coloneqq r \sim \mathcal{U}\left[ 0,1 \right]$\\
        \uIf{$S_i = 1$ {\bf and} $A_i = -1$ {\bf and} $r < a_1$}{
            $A_i = 1$\\
        }
        \ElseIf{$S_i = -1$ {\bf and} $A_i = 1$ {\bf and} $r < a_2$}{
            $A_i = -1$\\
        }
    }
}
\end{algorithm}

To clarify, we do the following. At the beginning of each independent simulation (trajectory), we set the initial conditions. In simulations, we always start from a fully unadopted, negative state, i.e., $\forall_i \,\, A_i(0)=-1, \,\, S_i(0)=-1$. Then, we perform Monte Carlo steps until a determined time horizon $T$ is reached ($T$ is the number of MCS). We repeat independent simulations multiple times with the same set of parameter values for a better statistical accuracy. We also run separate simulations for different sets of parameter values. 

To examine the model on a macroscopic scale, we use two measures: concentration (fraction) of positive adoption states and concentration of positive opinions.
\begin{definition}
    Let $N$ be the number of agents, $A_i \pm 1$ the adoption state and $S_i = \pm 1$ the opinion of agent $i$, for $i=1,2,\dots,N$. Then, the concentrations of positive adoption states $c_A$ and opinions $c_S$ are given by:
    \begin{equation}
        c_A = \frac{1}{2N}\sum_{i=1}^N \left( A_i + 1 \right), \,\,\,\, c_S = \frac{1}{2N}\sum_{i=1}^N \left( S_i + 1 \right).
        \label{eq:concentration}
    \end{equation}
    \label{def:concentration}
\end{definition}
\noindent By definition $c_A \in [0,1]$ and $c_S \in [0,1]$. For the sake of simplicity, we refer to them as just concentrations of adoption states (or adopted) and opinions, respectively. 

\subsection{Mean-Field Approximation}
\label{ssec:model_mfa}

Here, we utilize the mean-field approximation (MFA)~\cite{NYC12} to derive a set of equations describing the dynamical system.
\begin{theorem}
    Let $c_A$ denote the concentration of positive adoption states and the probability of choosing an agent with a positive adoption state, and $c_S$ -- the concentration of positive opinions and the probability of choosing an agent with a positive opinion. Under the assumptions that these events are independent, and that each layer of the network is of size $N \to \infty$, the dynamics of the system is described in the AND variant by:
    \begin{align}
        \frac{\text{d} c_A}{\text{d} t} &= c_S\left(1-c_A\right)a_1 - \left(1-c_S\right)c_Aha_1, \label{eq:pv_A} \\
        \frac{\text{d} c_S}{\text{d} t} &= \left(1-c_S\right) \left( \frac{1}{2}p + (1-p) c_S^q c_A^q \right) - c_S \left( \frac{1}{2}p + (1-p) (1-c_S)^q (1-c_A)^q \right), \label{eq:pv_Sand}
    \end{align}
    and in the OR variant by:
    \begin{align}
        \frac{\text{d} c_A}{\text{d} t} ={}& c_S\left(1-c_A\right)a_1 - \left(1-c_S\right)c_Aha_1, \nonumber\\
        \frac{\text{d} c_S}{\text{d} t} ={}& (1-c_S) \Bigl\{ \frac{1}{2}p + (1-p) \left( c_S^q \left( 1-c_A^q-(1-c_A)^q \right) \right. \nonumber \\
        & \left. + c_A^q \left( 1-c_S^q-(1-c_S)^q \right) + c_S^q c_A^q \right) \Bigr\} \nonumber \\
        & - c_S \Bigl\{ \frac{1}{2}p + (1-p) \left( (1-c_S)^q \left( 1-c_A^q-(1-c_A)^q \right) \right. \nonumber \\
        & \left. + (1-c_A)^q \left( 1-c_S^q-(1-c_S)^q \right) + (1-c_S)^q (1-c_A)^q \right) \Bigr\}.        
        \label{eq:pv_Sor}
    \end{align}
    \label{the:1}
\end{theorem}
\begin{proof}
    Under our assumptions, we can write down probability $\gamma_A^+$ that the number of agents with positive adoption states will increase by $1$ and probability $\gamma_A^-$ that it will decrease by~$1$:
    \begin{align}
        \gamma_A^+ &= c_S\left(1-c_A\right)a_1, \\
        \gamma_A^- &= \left(1-c_S\right)c_Aa_2.
    \end{align}
    Then, we derive the differential equation, analogously to \cite{KRU17,WER22}:
    \begin{equation}
        \frac{\text{d} c_A}{\text{d} t} = \gamma_A^+ - \gamma_A^-,
    \end{equation}
    which gives us Eq.~(\ref{eq:pv_A}), after replacing $a_2$ with $h a_1$, see Eq.~(\ref{eq:h}). Eq.~(\ref{eq:pv_A}) is common for both AND and OR variants of the model. Similarly, we can write down probabilities $\gamma_S^+$ and $\gamma_S^-$ that a number of agents with positive opinion will increase or decrease by $1$, respectively. These are:
    \begin{align}
        \gamma_{S,AND}^+ &= \left(1-c_S\right) \left( \frac{1}{2}p + (1-p) c_S^q c_A^q \right), \\
        \gamma_{S,AND}^- &= c_S \left( \frac{1}{2}p + (1-p) (1-c_S)^q (1-c_A)^q \right),
    \end{align}
    in the AND variant, and:
    \begin{align}
        \gamma_{S,OR}^+ ={}& (1-c_S) \left\{ \frac{1}{2}p + (1-p) \left( c_S^q \sum_{i=1}^{q-1} \binom{q}{i} c_A^i (1-c_A)^{q-i} \right.\right. \nonumber \\
    & \left.\left. + c_A^q \sum_{i=1}^{q-1} \binom{q}{i} c_S^i (1-c_S)^{q-i} + c_S^q c_A^q \right) \right\}, \\
        \gamma_{S,OR}^- ={}& c_S \left\{ \frac{1}{2}p + (1-p) \left( (1-c_S)^q \sum_{i=1}^{q-1} \binom{q}{i} (1-c_A)^i c_A^{q-i} \right.\right. \nonumber \\
    & \left.\left. + (1-c_A)^q \sum_{i=1}^{q-1} \binom{q}{i} (1-c_S)^i c_S^{q-i} + (1-c_S)^q (1-c_A)^q \right) \right\},
    \end{align}
    in the OR variant. Using the binomial formula, we can replace sums in $\gamma_{S,OR}^+$ and $\gamma_{S,OR}^-$ with:
    \begin{equation}
        \sum_{i=1}^{q-1} \binom{q}{i} c_X^i (1-c_X)^{q-i} = \sum_{i=1}^{q-1} \binom{q}{i} (1-c_X)^i c_X^{q-i} = 1 - (1-c_X)^q - c_X^q, \,\,\,\, X=A,S.
    \end{equation}
    Then:
    \begin{equation}
        \frac{\text{d} c_A}{\text{d} t} = \gamma_{S,AND}^+ - \gamma_{S,AND}^-
    \end{equation}
    gives us Eq.~(\ref{eq:pv_Sand}), and
    \begin{equation}
        \frac{\text{d} c_A}{\text{d} t} = \gamma_{S,OR}^+ - \gamma_{S,OR}^-
    \end{equation}
    gives Eq.~(\ref{eq:pv_Sor}).
\end{proof}

Due to the power of $q$ in Eqs.~(\ref{eq:pv_Sand})-(\ref{eq:pv_Sor}), time trajectories of the system given by Eqs.~(\ref{eq:pv_A})-(\ref{eq:pv_Sor}) cannot be determined analytically. Hence, one must resort to numerical methods, which we show in Section~\ref{sec:result}. However, stationary states can be found analytically and their existence can be proven.
\begin{theorem}
    Let $c_A$ denote the concentration of positive adoption states and the probability of choosing an agent with a positive adoption state, and $c_S$ -- the concentration of positive opinions and the probability of choosing an agent with a positive opinion. Under the assumptions that these events are independent, and that each layer of the network is of size $N \to \infty$, for any $q \in \mathbb{N}$, $q \geq 2$, $p \in [0,1]$, $h \in (0,1)$ and $a_1 \in (0,1]$, there always exists at least one stationary state, and all stationary states must satisfy:
    \begin{align}
        c_A &= \frac{c_S}{c_S+h-c_Sh}, \label{eq:stat_A}\\
        p &= \frac{f_2(c_S)}{f_2(c_S) + f_3(c_S)},
    \end{align}
    in the AND variant, and:
    \begin{align}
        c_A &= \frac{c_S}{c_S+h-c_Sh}, \nonumber \\
        p &= \frac{f_5(c_S)}{f_5(c_S) + f_3(c_S)}, \label{eq:stat_Sor}
    \end{align}
    in the OR variant. Here:
    \begin{align*}
        f_2(c_S) ={}& c_S(1-c_S) \left( c_S^{2q-1} - h^q (1-c_S)^{2q-1} \right), \\
        f_3(c_S) ={}& \left( c_S - \frac{1}{2} \right)(c_S+h-c_Sh)^q, \\
        f_5(c_S) ={}& c_S(1-c_S) \Bigl\{ \left( c_S^{q-1} - (1 - c_S)^{q-1} \right)(c_S+h-c_Sh)^q \\
        & + c_S^{q-1}(1-c_S)^{q-1}(1+h^q)(2c_S-1) \\
        & + c_S^{q-1} - c_S^{2q-1} + h^q(1-c_S)^{2q-1} - h^q(1-c_S)^{q-1} \Bigr\}.
    \end{align*}
\end{theorem}
\begin{proof}
    A stationary state must satisfy:
    \begin{equation}
        \frac{\text{d} c_A}{\text{d} t} = 0, \,\,\,\, \frac{\text{d} c_S}{\text{d} t} = 0.
    \end{equation}
    First, let us consider Eq.~(\ref{eq:pv_A}), common for both AND and OR variants of the model:
    \begin{equation}
        0 = c_S(1-c_A)a_1 - (1-c_S)c_Aha_1.
    \end{equation}
    Since $a_1 > 0$, we can rewrite it as:
    \begin{equation}
        c_S = c_A(c_S+h-c_Sh).
    \end{equation}
    Because $c_S \in [0,1]$ and $h \in (0,1)$, $(c_S+h-c_Sh) > 0$ always, and hence:
    \begin{equation}
        c_A = \frac{c_S}{c_S+h-c_Sh} = \frac{c_S}{(1-h)c_S+h} = f_1(c_S). \label{eq:f1}
    \end{equation}
    Note, that function $f_1(c_S)$ is a rational one of degree $1$. Therefore, for any $c_S \in [0,1]$, there always exists exactly one $c_A \in [0,1]$ that satisfies Eq.~(\ref{eq:f1}).

    Now, let us consider Eq.~(\ref{eq:pv_Sand}) from the AND variant:
    \begin{equation}
        0 = \left(1-c_S\right) \left( \frac{1}{2}p + (1-p) c_S^q c_A^q \right) - c_S \left( \frac{1}{2}p + (1-p) (1-c_S)^q (1-c_A)^q \right).
    \end{equation}
    By expanding the brackets and putting all the terms containing $p$ on one side, we arrive at:
    \begin{align}
        p \left( c_S^q c_A^q - c_S^{q+1}c_A^q - c_S(1-c_S)^q(1-c_A)^q + c_S - \frac{1}{2} \right) = \nonumber \\
        c_S^q c_A^q - c_S^{q+1}c_A^q - c_S(1-c_S)^q(1-c_A)^q.
    \end{align}
    Next, by replacing $c_A$ with $f_1(c_S)$ and multiplying all the terms by $(c_S+h-c_Sh)^q$, we get:
    \begin{align}
        p \left\{ c_S(1-c_S) \left( c_S^{2q-1} - h^q (1-c_S)^{2q-1} \right) + \left( c_S - \frac{1}{2} \right)(c_S+h-c_Sh)^q \right\} = \nonumber \\
        c_S(1-c_S) \left( c_S^{2q-1} - h^q (1-c_S)^{2q-1} \right).
    \end{align}
    Then, denoting:
    \begin{align}
        f_2(c_S) &= c_S(1-c_S) \left( c_S^{2q-1} - h^q (1-c_S)^{2q-1} \right), \label{eq:f2} \\
        f_3(c_S) &= \left( c_S - \frac{1}{2} \right)(c_S+h-c_Sh)^q, \label{eq:f3}
    \end{align}
    we obtain:
    \begin{equation}
        p \left( f_2(c_S) + f_3(c_S) \right) = f_2(c_S). \label{eq:f2f3}
    \end{equation}
    If $f_2(c_S) + f_3(c_S)=0$, then $f_2(c_S)=0$, and hence $f_3(c_S)=0$ as well. Function $f_3(c_S)=0$ only if $c_S=\frac{1}{2}$, as $(c_S+h-c_Sh) > 0$. For $c_S \in \left[ \frac{1}{2},1 \right)$, function $f_2(c_S)$:
    \begin{equation}
        f_2(c_S) = \underbrace{c_S(1-c_S)}_{> 0} \underbrace{\left( c_S^{2q-1} - h^q (1-c_S)^{2q-1} \right)}_{> 0} > 0, \label{eq:f205}
    \end{equation}
    because for $c_S \in \left[ \frac{1}{2},1 \right)$, $c_S \geq 1-c_S$, and $h \in (0,1)$. Hence, there is no $c_S$ for which $f_2(c_S)=f_3(c_S)=0$, and therefore no $c_S$ for which $f_2(c_S) + f_3(c_S)=0$ that satisfies Eq.~(\ref{eq:f2f3}). As further we consider only $f_2(c_S) + f_3(c_S) \neq 0$, we can write:
    \begin{equation}
        p = \frac{f_2(c_S)}{f_2(c_S) + f_3(c_S)} = f_4(c_S). \label{eq:f4}
    \end{equation}
    From Eqs.~(\ref{eq:f2})-(\ref{eq:f4}), we can notice that:
    \begin{align}
        &\text{for} \,\,\,\, c_S = 1: f_2(c_S) = 0, \,\, f_3(c_S) > 0 \implies f_4(c_S)=0 \nonumber \\
        &\text{for} \,\,\,\, c_S = \frac{1}{2}: f_2(c_S) > 0, \,\, f_3(c_S)=0 \implies f_4(c_S)=1, \nonumber \\
        &\text{for} \,\,\,\, c_S \in \left( \frac{1}{2}, 1 \right): f_2(c_S) > 0, \,\, f_3(c_S) > 0 \implies f_4(c_S) \in (0,1).
    \end{align}    
    When $c_S \in \left[ \frac{1}{2}, 1 \right]$, $f_4(c_S)$ is a rational function with a positive denominator, and hence continuous within that interval. As it is continuous in $\left[ \frac{1}{2},1 \right]$, $f_4(1)=0$ and $f_4\left(\frac{1}{2}\right)=1$, it achieves any value between $0$ and $1$. Hence, for any $p \in [0,1]$ there always exists at least one $c_S \in [0,1]$ that satisfies Eq.~(\ref{eq:f4}). Therefore, in the AND variant, there always exists at least one stationary state. 

    Similarly in the OR variant, we have:
    \begin{align}
        0 ={}& (1-c_S) \Bigl\{ \frac{1}{2}p + (1-p) \left( c_S^q \left( 1-c_A^q-(1-c_A)^q \right) \right. \nonumber \\
        & \left. + c_A^q \left( 1-c_S^q-(1-c_S)^q \right) + c_S^q c_A^q \right) \Bigr\} \nonumber \\
        & - c_S \Bigl\{ \frac{1}{2}p + (1-p) \left( (1-c_S)^q \left( 1-c_A^q-(1-c_A)^q \right) \right. \nonumber \\
        & \left. + (1-c_A)^q \left( 1-c_S^q-(1-c_S)^q \right) + (1-c_S)^q (1-c_A)^q \right) \Bigr\},
    \end{align}  
    from Eq.~(\ref{eq:pv_Sor}). Analogously to the AND variant, by expanding the brackets, putting all the terms containing $p$ on one side, then, by replacing $c_A$ with $f_1(c_S)$ and multiplying all the terms by $(c_S+h-c_Sh)^q$, we get:
    \begin{equation}
        p \left( f_5(c_S) + f_3(c_S) \right) = f_5(c_S), \label{eq:f5f3}
    \end{equation}
    where $f_3(c_S)$ is given by Eq.~(\ref{eq:f3}) and
    \begin{align}
        f_5(c_S) ={}& c_S(1-c_S) \Bigl\{ \left( c_S^{q-1} - (1 - c_S)^{q-1} \right)(c_S+h-c_Sh)^q \nonumber \\
        & + c_S^{q-1}(1-c_S)^{q-1}(1+h^q)(2c_S-1) \nonumber \\
        & + c_S^{q-1} - c_S^{2q-1} + h^q(1-c_S)^{2q-1} - h^q(1-c_S)^{q-1} \Bigr\}.
    \end{align}
    Again, if $f_5(c_S) + f_3(c_S)=0$, then $f_5(c_S)=0$, and hence $f_3(c_S)=0$ as well. Function $f_3(c_S)=0$ only if $c_S=\frac{1}{2}$. For $c_S \in \left[ \frac{1}{2},1 \right)$, function $f_5(c_S)$:
    \begin{align}
        f_5(c_S) ={}& \underbrace{c_S(1-c_S)}_{> 0} \Bigl\{ \underbrace{\left( c_S^{q-1} - (1 - c_S)^{q-1} \right)}_{\geq 0} \underbrace{(c_S+h-c_Sh)^q}_{> 0} \nonumber \\
        & + \underbrace{c_S^{q-1}(1-c_S)^{q-1}(1+h^q)(2c_S-1)}_{\geq 0} \nonumber \\
        & + \underbrace{c_S^{q-1} - c_S^{2q-1} + h^q(1-c_S)^{2q-1} - h^q(1-c_S)^{q-1}}_{> 0, \,\, \text{by Lemma \ref{lemma:1}}} \Bigr\} > 0,
    \end{align}
    because for $c_S \in \left[ \frac{1}{2},1 \right)$, $c_S \geq 1-c_S$, and $h \in (0,1)$. Lemma~\ref{lemma:1} is presented below this proof. Hence, there is no $c_S$ for which $f_5(c_S)=f_3(c_S)=0$, and therefore no $c_S$ for which $f_5(c_S) + f_3(c_S)=0$ that satisfies Eq.~(\ref{eq:f5f3}). As further we consider only $f_5(c_S) + f_3(c_S) \neq 0$, we can write:
    \begin{equation}
        p = \frac{f_5(c_S)}{f_5(c_S) + f_3(c_S)} = f_6(c_S). \label{eq:f6}
    \end{equation}
    From Eqs.~(\ref{eq:f5f3})-(\ref{eq:f6}), we can notice that:
    \begin{align}
        &\text{for} \,\,\,\, c_S = 1: f_5(c_S) = 0, \,\, f_3(c_S) > 0 \implies f_6(c_S)=0 \nonumber \\
        &\text{for} \,\,\,\, c_S = \frac{1}{2}: f_5(c_S) > 0, \,\, f_3(c_S)=0 \implies f_6(c_S)=1, \nonumber \\
        &\text{for} \,\,\,\, c_S \in \left( \frac{1}{2}, 1 \right): f_5(c_S) > 0, \,\, f_3(c_S) > 0 \implies f_6(c_S) \in (0,1).
    \end{align}    
    When $c_S \in \left[ \frac{1}{2}, 1 \right]$, $f_6(c_S)$ is a rational function with positive denominator, and hence continuous within that interval. As it is continuous in $\left[ \frac{1}{2},1 \right]$, $f_6(1)=0$ and $f_6\left(\frac{1}{2}\right)=1$, it achieves any value between $0$ and $1$. Hence, for any $p \in [0,1]$ there always exists at least one $c_S \in [0,1]$ that satisfies Eq.~(\ref{eq:f6}). Therefore, in the OR variant, there always exists at least one stationary state. 
\end{proof}
\begin{lemma}
    Let $c_S \in \left[ \frac{1}{2},1 \right)$, $h \in (0,1)$ and $q \in \mathbb{N}$, $q \geq 2$. Then:
    \begin{equation}
        c_S^{q-1} - c_S^{2q-1} + h^q(1-c_S)^{2q-1} - h^q(1-c_S)^{q-1} > 0.
    \end{equation}
    \label{lemma:1}
\end{lemma}
\begin{proof}
    First, we note that for $c_S \in \left[ \frac{1}{2},1 \right)$:
    \begin{align}
        c_S^{2q-2} - (1-c_S)^{2q-2} &\geq c_S^{2q-1} - (1-c_S)^{2q-1} \nonumber \\
        c_S^{2q-2} - c_S^{2q-1} &\geq (1-c_S)^{2q-2} - (1-c_S)^{2q-1} \nonumber \\
        c_S^{2q-2}(1-c_S) &\geq (1-c_S)^{2q-2}(1-(1-c_S)) \nonumber \\
        c_S^{2q-3} &\geq (1-c_S)^{2q-3} \label{eq:2q3}
    \end{align}
    is true, because $c_S \geq 1-c_S$ and $2q-3 \geq 1$. Then:
    \begin{align}
        & c_S^{q-1} - c_S^{2q-1} + h^q(1-c_S)^{2q-1} - h^q(1-c_S)^{q-1} = \nonumber \\
        & c_S^{q-1} - c_S^{2q-1} - h^q \underbrace{\left( (1-c_S)^{q-1} - (1-c_S)^{2q-1} \right)}_{> 0} \stackrel{h < 1}{>} \nonumber \\
        & c_S^{q-1} - c_S^{2q-1} - \left( (1-c_S)^{q-1} - (1-c_S)^{2q-1} \right) = \nonumber \\
        & c_S^{q-1} - (1-c_S)^{q-1} - \left( c_S^{2q-1} - (1-c_S)^{2q-1} \right) \stackrel{\text{Eq.~(\ref{eq:2q3})}}{\geq} \nonumber \\
        & c_S^{q-1} - (1-c_S)^{q-1} - \left( c_S^{2q-2} - (1-c_S)^{2q-2} \right) = \nonumber \\
        & c_S^{q-1} - (1-c_S)^{q-1} - \left( c_S^{q-1} - (1-c_S)^{q-1} \right) \left( c_S^{q-1} + (1-c_S)^{q-1} \right) = \nonumber \\
        & \underbrace{\left( c_S^{q-1} - (1-c_S)^{q-1} \right)}_{\geq 0} \underbrace{\left( 1 - c_S^{q-1} - (1-c_S)^{q-1} \right)}_{\geq 0} \geq 0.
    \end{align}
\end{proof}

\section{Results}
\label{sec:result}

\subsection{Simulation Details}
\label{ssec:res_dim_det}

First, we examine our model with Monte Carlo computer simulations. In order to measure the outcomes at the macroscopic level, we use concentrations, i.e. fractions, of positive adoption states $c_A$ and positive opinions $c_S$, see Definition~\ref{def:concentration}. In all the simulations, we study a system of size $N=2500$ (number of agents). 
As mentioned in the Introduction, the first layer of the two-layer network is always a Square Lattice with Moore's neighborhood \cite{MOO62}, and no periodic boundary conditions, which means that almost every agent has exactly 8 neighbors (except those on the edges of the lattice). We denote it by SL($N=2500$, $m=1$), where $m=1$ indicates the range of a neighborhood on a square grid. The second layer is always a two-dimensional Watts-Strogatz graph \cite{BAT22}, i.e. a Square Lattice with Moore's neighborhood before rewiring and rewiring probability $\beta=0.2$, and again, without periodic boundary conditions. 
We denote it by WS2D($N=2500$, $m=1$, $\beta=0.2$). We study the behavior of the system with respect to the parameters $p$, $a_1$ and $a_2$. 
%
%
Although the size of the group of influence $q$ is a~parameter as well, here we keep it constant, as its impact in the $q$-voter model with independence is well studied \cite{NYC12}. We choose $q=4$ motivated by real life experiments~\cite{ASC55}. This will help us to reduce computational efforts. In all the simulations, we start from a fully unadopted, negative state, i.e., $c_A(0)=0$, $c_S(0)=0$. Lastly, for illustrative purposes, the time horizon $T$ for each simulation is limited to $5000$ Monte Carlo steps (MCS), and the number of time trajectories to $10$. 

\subsection{Simulation Model}
\label{ssec:res_sim}

\begin{figure}
    \centering
    \includegraphics[width=0.6\textwidth]{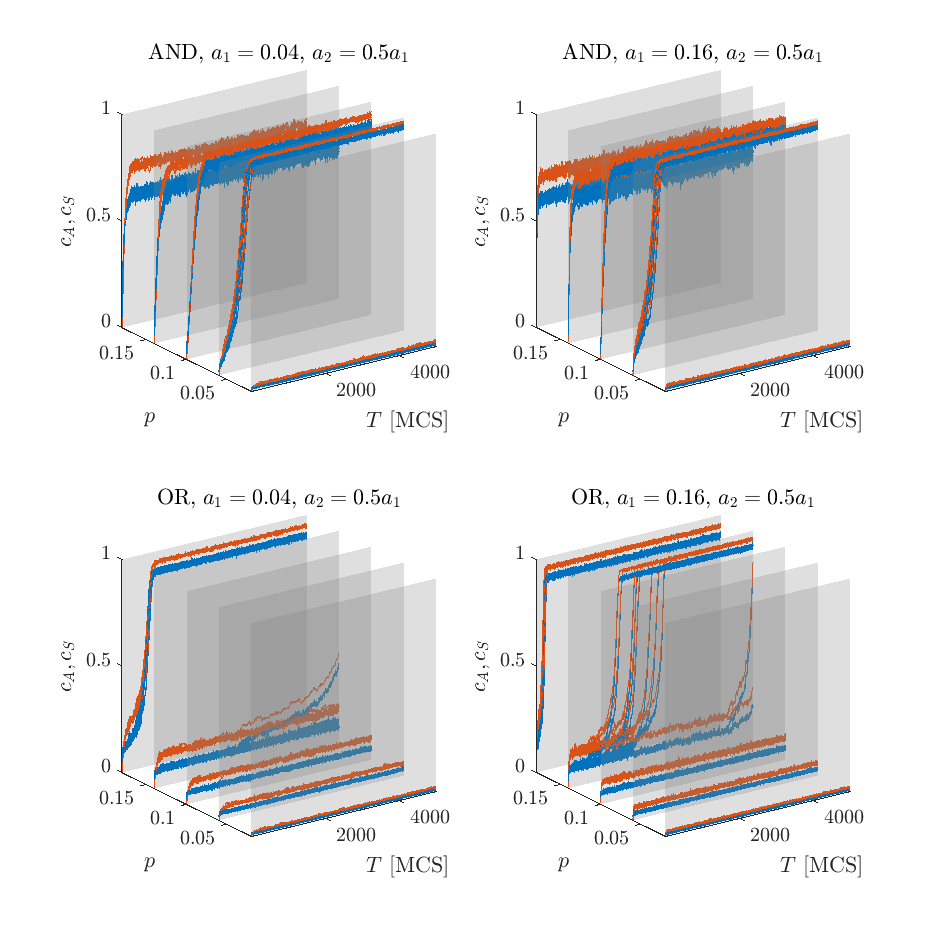}
    \caption{10 simulated time trajectories of $c_A$ (red) and $c_S$ (blue) for different values of $p$: the AND (top), the OR variant (bottom). Values of $a_1$: $a_1=0.04$ (left), $a_1=0.16$ (right). First layer -- SL(N,1), second layer -- WS2D(N,1,0.2), size $N=2500$ and $a_2=0.5a_1$ in all cases.}
    \label{fig:sim2}
\end{figure}

Let us begin with Fig.~\ref{fig:sim2}. It shows sample evolutions of the system, for both variants (AND and OR), and a range of values of $p$, $a_1$ and $a_2$, where $a_2=0.5a_1$ in all cases. We can observe that several final outcomes are possible. When $p$ is low, we retain an unadopted state ($c_A(T) \approx 0$, $c_S(T) \approx 0$). None of the agents have a positive opinion or adoption state, except for a few ``rebels''. Then, there is a critical value of $p$, above which the system becomes adopted ($c_A(T) \approx 1$, $c_S(T) \approx 1$). If the value of $p$ is high, independence surpasses conformity and we end up with a disordered system ($c_A(T) \approx 0.5$, $c_S(T) \approx 0.5$). There is a visible difference between the two variants: AND and OR. In the first, much lower values of $p$ are required to enter adopted or disordered state than in the latter. Finally, the value of $a_1$ itself has no impact on the final state (except for $a_1=0$, which is not shown here), only on the time needed to reach it, with the higher value speeding up the process. 

\begin{figure}
    \centering
    \includegraphics[width=0.6\textwidth]{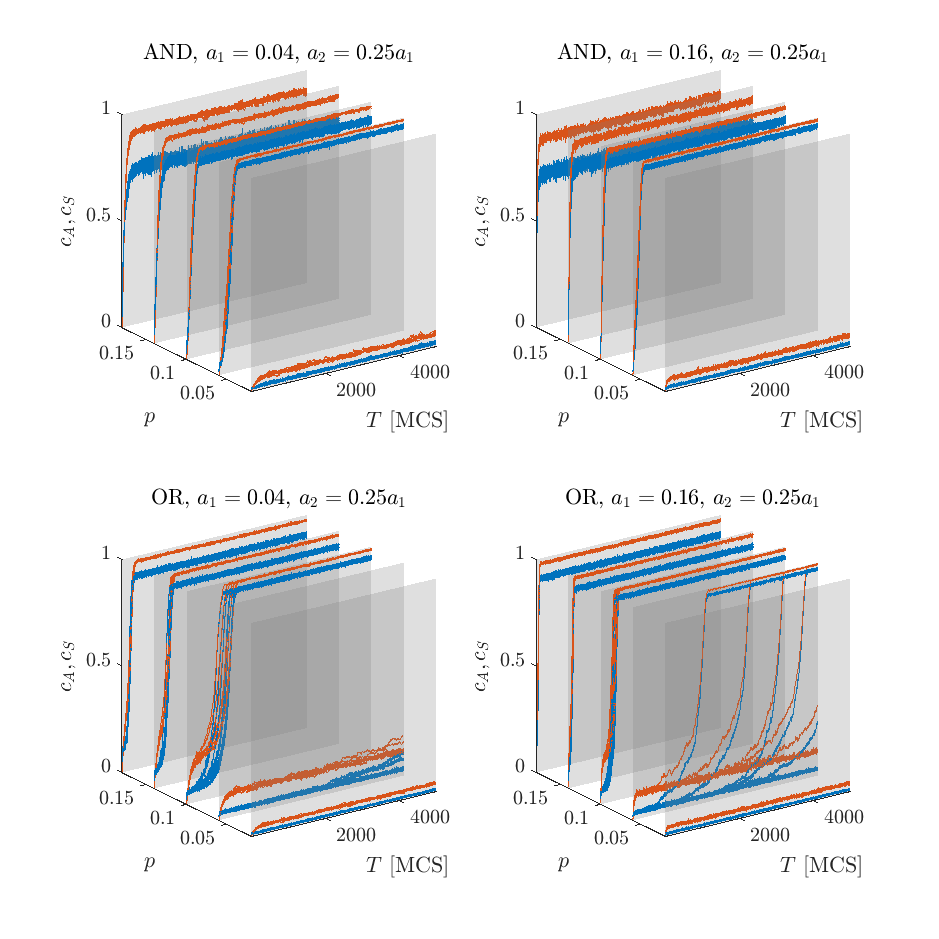}
    \caption{10 simulated time trajectories of $c_A$ (red) and $c_S$ (blue) for different values of $p$: the AND (top), the OR variant (bottom). Values of $a_1$: $a_1=0.04$ (left), $a_1=0.16$ (right). First layer -- SL(N,1), second layer -- WS2D(N,1,0.2), size $N=2500$ and $a_2=0.25a_1$ in all cases.}
    \label{fig:sim4}
\end{figure}

What affects the final state, however, is the relationship between $a_1$ and $a_2$, namely the coefficient $h$, see Eq.~(\ref{eq:h}). In Fig.~\ref{fig:sim4}, $a_2$ is twice smaller than before, i.e. $a_2=0.25a_1$. With this decrease in $a_2$ much lower values of $p$ are needed, in both variants, for the system to become adopted. This itself is a rather trivial conclusion, as now the probability of loosing adoption is four times smaller that the one of adopting. There is, however, a secondary effect to it, due to the fact how the first layer of the network impacts opinions in our model. Positive adoption states support positive opinions (and vice versa). For this reason, the adopted state is not only easier to reach, but also more difficult to disorder. 

We decided not to examine a symmetrical case $a_1=a_2$, i.e. $h=1$, because there is no adopted state there. From $c_A(0)=0$, $c_S(0)=0$, the system can only evolve into total disorder or remain unadopted.


\subsection{Mean-Field Approximation}
\label{ssec:res_mfa}

Monte Carlo computer simulations are the first choice when it comes to agent-based models. Unfortunately, as mentioned in the Introduction, they can be very time-consuming. For this reason, in Theorem~\ref{the:1} we derived a set of equations describing dynamics of the system, see Eqs.~(\ref{eq:pv_A})-(\ref{eq:pv_Sor}), under certain assumptions. Thanks to this approach, we can examine a wide range of parameter values, which we could not achieve with Monte Carlo simulations in any reasonable amount of time.

\begin{figure}
    \centering
    \includegraphics[width=0.6\textwidth]{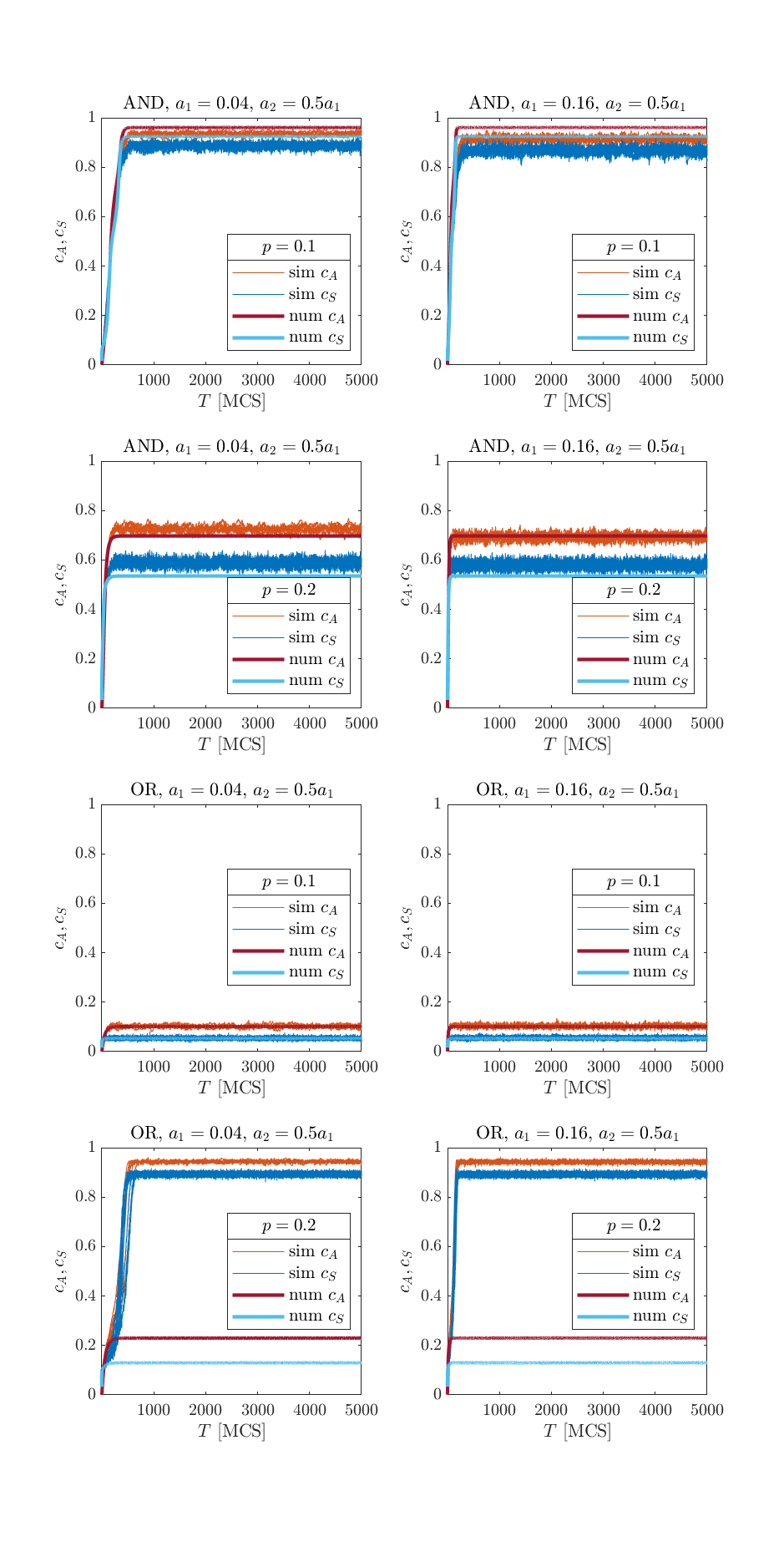}
    \caption{10 simulated time trajectories of $c_A$ and $c_S$ versus numerically obtained time trajectory from Eqs.~(\ref{eq:pv_A})-(\ref{eq:pv_Sor}), for different values of $p$ and $a_1$. The AND (top 4) and the OR variant (bottom 4). First layer -- SL(N,1), second layer -- WS2D(N,1,0.2), size $N=2500$ and $a_2=0.5a_1$ in all cases.}
    \label{fig:sim_num}
\end{figure}

Before discussing the analytical results themselves, we compare them with the simulations. In Fig.~\ref{fig:sim_num}, we combine the latter and numerically obtained time trajectories from Eqs.~(\ref{eq:pv_A})-(\ref{eq:pv_Sor}). As shown there, mean-field gives a fairly good approximation, except for the bottom row (OR variant, $p=0.2$). This is due to the fact, that in the MFA critical values of $p$ required for the system to adopt are slightly higher, as presented in Figs.~\ref{fig:anal2} and \ref{fig:anal4}.

\begin{figure}
    \centering
    \includegraphics[width=0.6\textwidth]{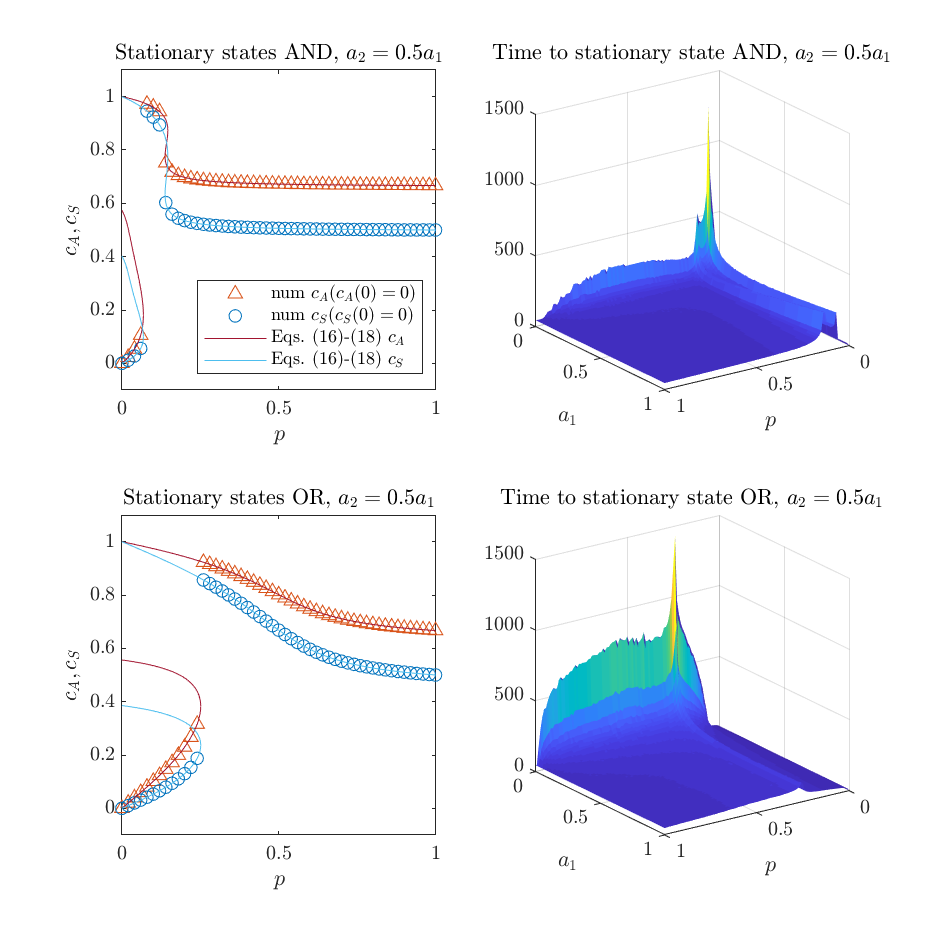}
    \caption{Stationary states (left) and time to reach a stationary state (right) for the AND (top) and the OR variant (bottom). Left side compares numerical results (from Eqs.~(\ref{eq:pv_A})-(\ref{eq:pv_Sor}); markers) vs analytical (from Eqs.~(\ref{eq:stat_A})-(\ref{eq:stat_Sor}); continuous lines). Numerical results cover only a portion of analytical ones, as they present stationary states from a single pair of initial conditions $(c_A(0)=0, c_S(0)=0)$ only, while the latter show all the possible stationary states. Right side shows time to reach a stationary state obtained with numerical methods. Adoption probabilities $a_1=0.5$ (left side only) and $a_2=0.5a_1$.}
    \label{fig:anal2}
\end{figure}

Mean-field time trajectories, and hence times to reach a stationary state, are obtained numerically from Eqs.~(\ref{eq:pv_A})-(\ref{eq:pv_Sor}). Analytically, we are only able to compute stationary states, see Eqs.~(\ref{eq:stat_A})-(\ref{eq:stat_Sor}). We compare the two in Fig.~\ref{fig:anal2}, for $a_2=0.5a_1$. Numerically computed stationary states perfectly match analytical ones. Segments not covered by the numerical results are due to the fact that the analytical solution shows all the possible stationary states, while the numerical one only those achievable from given initial conditions ($c_A(0)=0$, $c_S(0)=0$). Clearly visible here are the three possible groups of stationary states: unadopted ($c_A(T) \approx 0$, $c_S(T) \approx 0$), adopted ($c_A(T) \approx 1$, $c_S(T) \approx 1$) and disordered ($c_A(T) \approx 0.5$, $c_S(T) \approx 0.5$), although the transition between the latter two in the OR variant is very smooth. As already mentioned, the value of $a_1 \in (0,1]$ itself has no impact on a stationary state. Therefore, we set it arbitrarily to $0.5$ (only the left side of Fig.~\ref{fig:anal2}). On the right, times to reach a stationary state with respect to $a_1$ and $p$ are presented. These drop dramatically with an increase of $a_1$, but only for very low values of $a_1$. After that, the change is unnoticeable. There are significant ``ridges'' for the values of $p$ corresponding to transitions between groups of stationary states (left side). There are two such ridges in the AND variant, but only one in the OR variant, as the transition between adopted and disordered states is very smooth there. These increases in times are logical, as the system needs more time to ``decide'' which path to take, and consistent with our knowledge on phase transitions~\cite{WER22a}.

\begin{figure}
    \centering
    \includegraphics[width=0.6\textwidth]{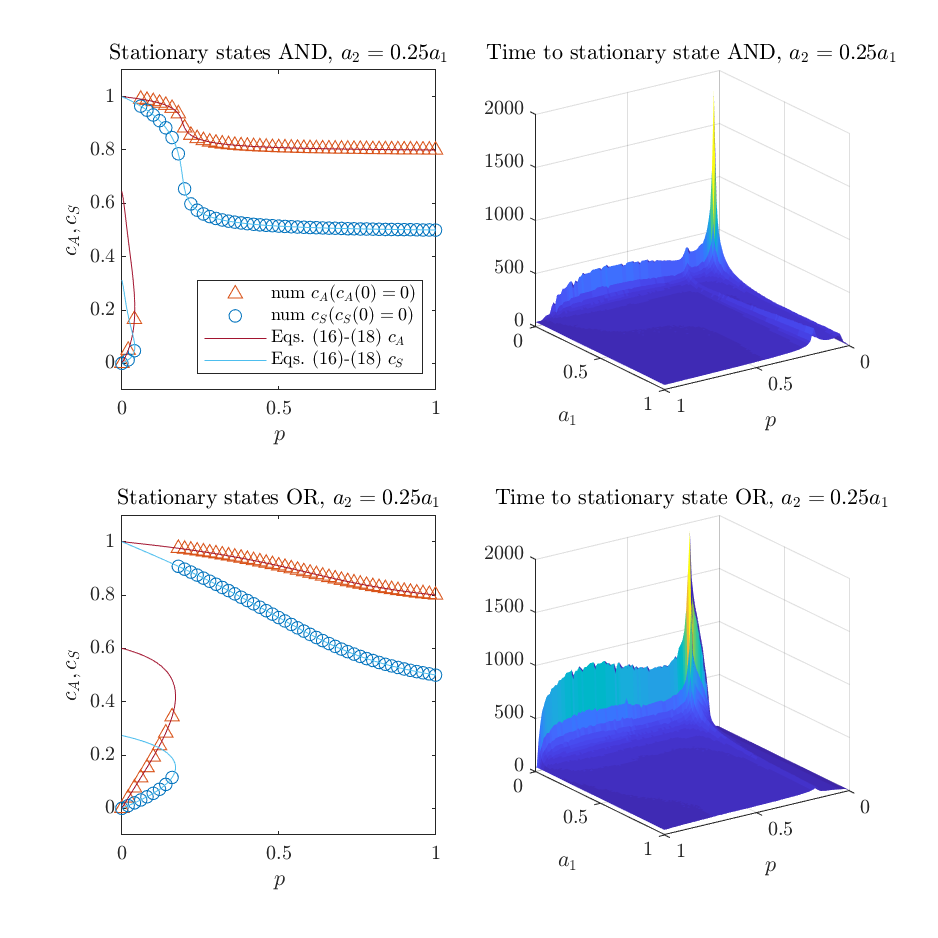}
    \caption{Stationary states (left) and time to reach a stationary state (right) for the AND (top) and the OR variant (bottom). Left side compares numerical results (from Eqs.~(\ref{eq:pv_A})-(\ref{eq:pv_Sor}); markers) vs analytical (from Eqs.~(\ref{eq:stat_A})-(\ref{eq:stat_Sor}); continuous lines). Numerical results cover only a portion of analytical ones, as they present stationary states from a single pair of initial conditions $(c_A(0)=0, c_S(0)=0)$ only, while the latter show all the possible stationary states. Right side shows time to reach a stationary state obtained with numerical methods. Adoption probabilities $a_1=0.5$ (left side only) and $a_2=0.25a_1$.}
    \label{fig:anal4}
\end{figure}

When we decrease $a_2$ to $0.25a_1$ (see Fig.~\ref{fig:anal4}), transitions between adopted and disordered states become less apparent in both variants of the model. Moreover, values of $p$ needed for adoption decrease, while those needed for disorder increase, which is consistent with the simulation results.

\section{Conclusions}
\label{sec:conclusion}

In this research, we proposed a new agent-based model that describes both opinion formation and diffusion of photovoltaic panels. Although, it is based on the well-known model of binary opinion dynamics, the $q$-voter model, we extended it by adding a second agent attribute (adoption state) and rewrote it to a two-layer network structure. To generalize the model from a single-layer, we considered two variants. First, when the $q$-voter's social influence is required on both layers of the network to impact one's opinion (the AND variant). Second, where influence from just one layer is sufficient (the OR variant). We investigated the model, the effect of parameters on stationary states and times to reach them, using both Monte Carlo computer simulations and the mean-field approximation. For the approximation's results, we used analytical methods wherever it was possible, and numerical ones otherwise.

Let us begin with our main goal, which was to build a diffusion of innovation model that 
does not reduce the dynamics to a one-dimensional mechanism, as classical models do \cite{BAS69}, but combines the adoption process with opinion formation. Consequently, diffusion of innovation and opinion dynamics are dependent on each other, with parameters only impacting one of these directly, influencing the other as well.

More specifically, within the model there exist three possible final states: unadopted system, adopted and finally, disordered one. The system transitions from the first, through the second, to the third, as the probability of independence $p$ increases. Broadly speaking, independence is initially essential for the adoption process to take off. However, as the diffusion accelerates, independence hinders innovation and ultimately prevents full adoption. Except for $a_1=0$, the probability of adoption has no impact on the stationary states, only on times to reach them. What impacts stationary states, however, is the relation between $a_1$ and probability of unadopting $a_2$. As $a_2$ becomes much smaller than $a_1$, critical values of $p$ required to reach the adopted state decrease, while those required for the disordered one increase. Basically, the greater the difference between people's willingness to install solar panels versus their willingness to get rid of them, the better for the innovation. This difference not only amplifies independence in the initial phase of adoption, but dilutes it in the final phase as well. 

Regarding the studied variants, there is a significant difference between the two, AND and OR. In the AND variant, adopted and disordered states are achieved more easily than in the OR variant, i.e. lower values of $p$ are sufficient, but the transitions between these states are more pronounced. This is because, in the AND variant the initial negative consensus is more easily disrupted than in the OR variant, but so is the later positive consensus.

Finally, it should be noted that MFA approximates the simulation results fairly well. Despite the fact that in the simulation version the network consists of a Square Lattice and a two-dimensional Watts-Strogatz graph, structures pretty different from complete graph. Still, there are visible differences. Critical values of $p$ to reach various states differ slightly, but the general picture remains.

\section*{Appendix}
\label{sec:appendix}

In the results presented above, only the $c_A(0)=0$, $c_S(0)=0$ initial state was considered. However, one could investigate others as well. Here, we examine, what happens if initially a small fraction of agents possess positive opinions and adoption states, $A(0)=+1$, $S(0)=+1$. These agents could be chosen randomly or determined by some centrality measure. We compare time trajectories for the basic case (no such agents), with the number of such agents chosen randomly and the same number determined by the highest degree on the second layer (degree centrality). 

\begin{figure}[ht]
    \centering
    \includegraphics[width=0.9\textwidth]{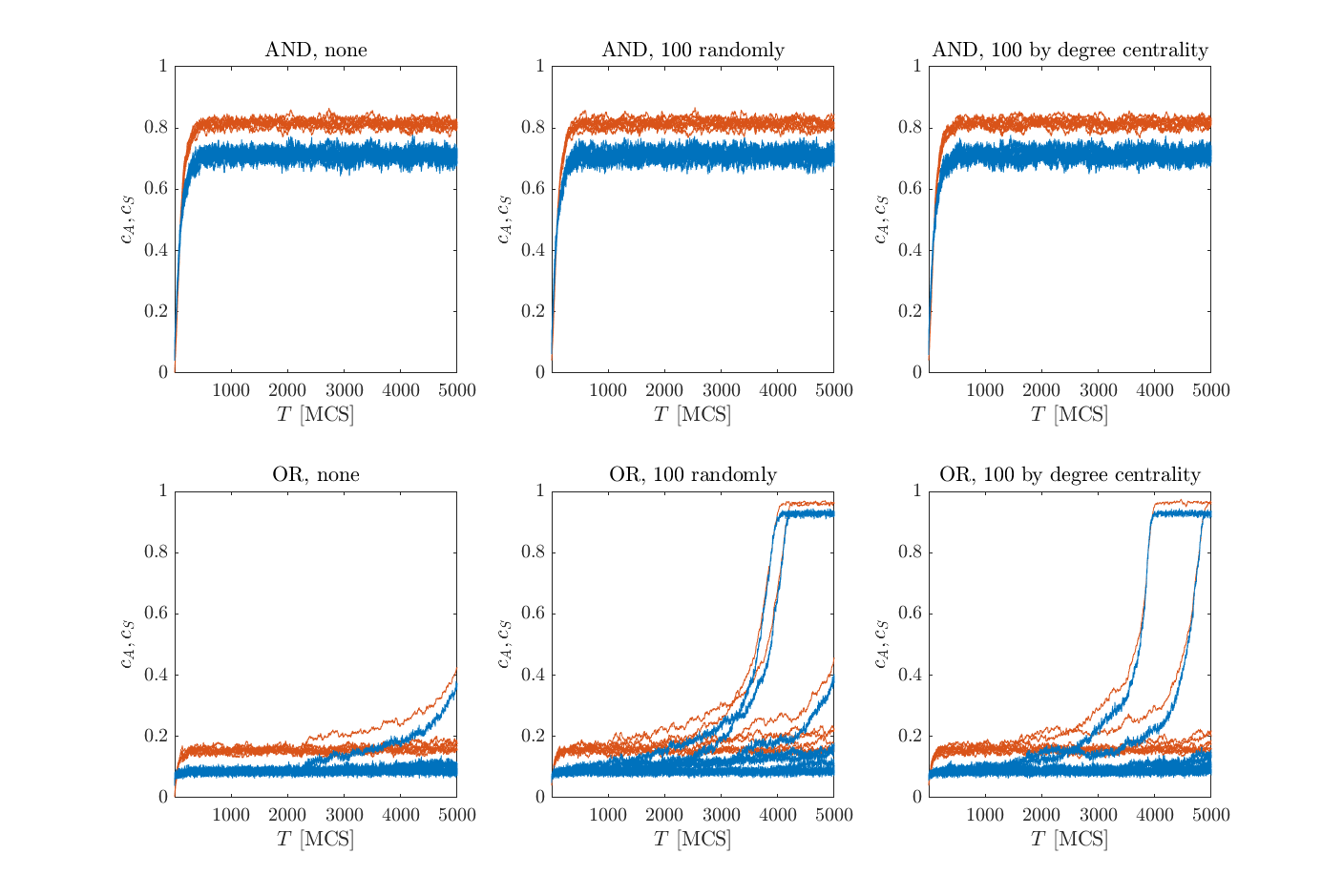}
    \caption{10 simulated time trajectories of $c_A$ (red) and $c_S$ (blue) for $c_A(0)=0$, $c_S(0)=0$ (left), 100 agents with $A(0)=+1$, $S(0)=+1$ chosen randomly (middle) and 100 agents with $A(0)=+1$, $S(0)=+1$ determined by the highest degree on the second layer (number of neighbors, right). The AND (top) and the OR variant (bottom). First layer -- SL(N,1), second layer -- WS2D(N,1,0.2), size $N=2500$ and $a_2=0.5a_1$ in all cases.}
    \label{fig:degree}
\end{figure}

Results are shown in Fig.~\ref{fig:degree}. In the AND variant there are no visible changes, as such agents easily loose either their positive opinion or adoption state when their numbers are low. Then, with the AND rule in play, they have negligible impact on the whole system. In the OR variant however, there are significant differences. Existence of such agents visibly decreases the time needed for the system to reach the adopted state. The method of their choice bears a lesser impact. This is due to the fact, that the second layer is the two-dimensional Watts-Strogatz graph, WS2D($N=2500$, $m=1$, $\beta=0.2$). Degree distribution of such a random graph corresponds to the binomial (for small networks, $N=10^2$) or the Poisson distribution (for larger networks, $N\geq 10^3$) \cite{BAR16}. Therefore, the network lacks agents with degrees greatly exceeding the average, who could bear an impact.

\section*{Acknowledgements}
This study was supported by the Ministry of Science and Higher Education, Poland, within the ‘‘Diamond Grant’’ Program through grant no. DI2019~0008~49.

\bibliographystyle{elsarticle-num} 
\bibliography{weron_szwabinski}





\end{document}